\documentclass[a4paper]{article}

\usepackage[english]{babel}
\usepackage{german}
\usepackage{amssymb}
\usepackage{amsfonts}
\usepackage{amsmath}
\usepackage{graphicx}
\usepackage{subfigure}
\usepackage{algorithm}
\usepackage[noend]{algorithmic}
\usepackage{amsmath,amssymb,cite}
\usepackage{color,hyperref}
\usepackage{xspace}
\usepackage{todonotes}
\usepackage{fancyhdr}

\newtheorem{theorem}{Theorem}

\newtheorem{corollary}{Corollary}

\newtheorem{definition}{Definition}

\newtheorem{lemma}{Lemma}

\newenvironment{proof}[1][Proof]{\noindent\textbf{#1.} }{\ \rule{0.5em}{0.5em}}

\begin{document}

\title{\vspace{-0.5cm}Binary Search in Graphs Revisited\thanks{Partially supported by the EPSRC grants EP/P020372/1, EP/P02002X/1, EP/L011018/1, and by the ISF grant 2021296.}}
\author{Argyrios Deligkas\thanks{Department of Computer Science, University of Liverpool, UK.
Email: \texttt{argyrios.deligkas@liverpool.ac.uk}} 
\and George B.~Mertzios\thanks{Department of Computer Science, Durham University, UK. 
Email: \texttt{george.mertzios@durham.ac.uk}} 
\and Paul G. Spirakis\thanks{Department of Computer Science, University of Liverpool, UK, and University of Patras, Greece. 
Email: \texttt{p.spirakis@liverpool.ac.uk}}}
\date{\vspace{-1.0cm}}
\maketitle

\begin{abstract}
In the classical binary search in a path the aim is to detect an unknown target by asking as few queries as possible, where each query reveals the direction to the target. This binary search algorithm has been recently extended by [Emamjomeh-Zadeh et al., \emph{STOC}, 2016] to the problem of detecting a target in an arbitrary graph. Similarly to the classical case in the path, the algorithm of Emamjomeh-Zadeh et al.~maintains a candidates' set for the target, while each query asks an appropriately chosen vertex-- the ``median''--which minimises a potential $\Phi$ among the vertices of the candidates' set. In this paper we address three open questions posed by Emamjomeh-Zadeh et al., namely (a)~detecting a target when the query response is a direction to an \emph{approximately shortest path} to the target, (b)~detecting a target when querying a vertex that is an \emph{approximate median} of the current candidates' set (instead of an exact one), and (c)~detecting \emph{multiple targets}, for which to the best of our knowledge no progress has been made so far. We resolve questions (a) and (b) by providing appropriate upper and lower bounds, as well as a new potential $\Gamma$ that guarantees efficient target detection even by querying an approximate median each time. With respect to (c), we initiate a systematic study for detecting two targets in graphs and we identify sufficient conditions on the queries that allow for strong (linear) lower bounds and strong (polylogarithmic) upper bounds for the number of queries. All of our positive results can be derived using our new potential $\Gamma$ that allows querying approximate medians.\newline

\textbf{Keywords:} binary search, graph, approximate query, probabilistic algorithm, lower bound.
\end{abstract}

\section{Introduction}

The classical binary search algorithm detects an unknown target (or
``treasure'') $t$ on a path with $n$
vertices by asking at most $\log n$ queries to an oracle which always
returns the direction from the queried vertex to $t$. To achieve this upper
bound on the number of queries, the algorithm maintains a set of candidates
for the place of $t$; this set is always a sub-path, and initially it is the
whole path. Then, at every iteration, the algorithm queries the middle
vertex (``median'') of this candidates' set
and, using the response of the query, it excludes either the left or the
right half of the set. This way of searching for a target in a path can be
naturally extended to the case where $t$ lies on an $n$-vertex tree, again
by asking at most $\log n$ queries that reveal the direction in the (unique)
path to $t$~\cite{OP06}. The principle of the binary search algorithm on
trees is based on the same idea as in the case of a path: for every tree
there exists a separator vertex such that each of its subtrees contains at
most half of the vertices of the tree~\cite{Jordan1869}, which can be also
efficiently computed.

Due to its prevalent nature in numerous applications, the problem of
detecting an unknown target in an arbitrary graph or, more generally in a
search space, has attracted many research attempts from different
viewpoints. Only recently the binary search algorithm with $\log n$
direction queries has been extended to arbitrary graphs by Emamjomeh-Zadeh
et al.~\cite{EKS16}. In this case there may exist multiple paths, or even
multiple shortest paths form the queried vertex to $t$. The direction query
considered in~\cite{EKS16} either returns that the queried vertex $q$ is the
sought target $t$, or it returns an arbitrary direction from $q$ to $t$,
i.e.~an arbitrary edge incident to $q$ which lies on a shortest path from $q$
to $t$. The main idea of this algorithm follows again the same principle as
for paths and trees: it always queries a vertex that is the
``median'' of the current candidates' set
and any response to the query is enough to shrink the size of the
candidates' set by a factor of at least $2$. Defining what the
``median'' is in the case of general graphs
now becomes more tricky: Emamjomeh-Zadeh et al.~\cite{EKS16} define the
median of a set $S$ as the vertex $q$ that minimizes a potential function $%
\Phi$, namely the sum of the distances from $q$ to all vertices of $S$.

Apart from searching for upper bounds on the number of queries needed to
detect a target $t$ in graphs, another point of interest is to derive
algorithms which, given a graph $G$, compute the \emph{optimal} number of
queries needed to detect an unknown target in $G$ (in the worst case). This
line of research was initiated in~\cite{LS85} where the authors studied
directed acyclic graphs (DAGs). Although computing a query-optimal algorithm
is known to be NP-hard~on general graphs~\cite{CDKL, D08, LY98}, there exist
efficient algorithms for trees; after a sequence of papers~\cite{IRV88,
LY01, S89, BFN99, MOW08}, linear time algorithms were found in~\cite{OP06,
MOW08}. Different models with queries of non-uniform costs or with a
probability distribution over the target locations were studied in~\cite%
{CJLM, CJLV, LMP, DKUZ}.

A different line of research is to search for upper bounds and
information-theoretic bounds on the number of queries needed to detect a
target $t$, assuming that the queries incorporate some degree of
``noise''. In one of the variations of this
model~\cite{EKS16,BH08,FRPU}, each query independently returns with
probability $p>\frac{1}{2}$ a direction to a shortest path from the queried
vertex $q$ to the target, and with probability $1-p$ an arbitrary edge
(possibly adversarially chosen) incident to $q$. The study of this problem
was initiated in~\cite{FRPU}, where $\Omega (\log n)$ and $O(\log n)$ bounds
on the number of queries were established for a path with $n$ vertices. This
information-theoretic lower bound of~\cite{FRPU} was matched by an improved
upper bound in~\cite{BH08}. The same matching bound was extended to general
graphs in~\cite{EKS16}.

In a further ``noisy'' variation of binary
search, every vertex $v$ of the graph is assigned a fixed edge incident to $%
v $ (also called the ``advice'' at $v$).
Then, for a fraction $p>\frac{1}{2}$ of the vertices, the advice directs to
a shortest path towards $t$, while for the rest of the vertices the advice
is arbitrary, i.e.~potentially misleading or adversarially chosen~\cite%
{BKR16, bo16}. This problem setting is motivated by the situation of a tourist
driving a car in an unknown country that was hit by a hurricane which
resulted in some fraction of road-signs being turned in an arbitrary and
unrecognizable way. The question now becomes whether it is still possible to
navigate through such a disturbed and misleading environment and to detect
the unknown target by asking only few queries (i.e.~taking advice only from
a few road-signs). It turns out that, apart from its obvious relevance to
data structure search, this problem also appears in artificial intelligence
as it can model searching using unreliable heuristics~\cite{BKR16,N71,P84}.
Moreover this problem also finds applications outside computer science, such
as in navigation issues in the context of collaborative transport by ants~\cite{Fonio16}.

Another way of incorporating some ``noise'' in the query responses, while trying to detect a target, is to have 
\emph{multiple targets} hidden in the graph. Even if there exist only two unknown
targets $t_{1}$ and $t_{2}$, the response of each query is potentially
confusing even if \emph{every} query correctly directs to a shortest path
from the queried vertex to one of the targets. The reason of confusion is
that now a detecting algorithm does not know to \emph{which} of the hidden
targets each query directs. In the context of the above example of a tourist
driving a car in an unknown country, imagine there are two main football
teams, each having its own stadium. A fraction $0<p_{1}<1$ of the population
supports the first team and a fraction $p_{2}=1-p_{1}$ the second one, while
the supporters of each team are evenly distributed across the country. The
driver can now ask questions of the type ``where is the
football stadium?'' to random local people along the way,
in an attempt to visit \emph{both} stadiums. Although every response will be
honest, the driver can never be sure which of the two stadiums the local
person meant. Can the tourist still detect both stadiums quickly enough? To
the best of our knowledge the problem of detecting multiple targets in
graphs has not been studied so far; this is one of the main topics of the
present paper.

The problem of detecting a target within a graph can be seen as a special
case of a two-player game introduced by Renyi~\cite{Renyi} and rediscovered
by Ulam~\cite{Ulam}. This game does not necessarily involve graphs: the
first player seeks to detect an element known to the second player in some
search space with $n$ elements. To this end, the first player may ask
arbitrary yes/no questions and the second player replies to them honestly or
not (according to the details of each specific model). Pelc~\cite{Pelc02}
gives a detailed taxonomy for this kind of games. \emph{Group testing} is a
sub-category of these games, where the aim is to detect all unknown objects
in a search space (not necessarily a graph)~\cite{DH93}. Thus, group testing
is related to the problem of detecting multiple targets in graphs, which we
study in this paper.

It is worth noting that techniques similar to~\cite{EKS16} were used to 
derive frameworks for robust interactive learning~\cite{EK17} and for 
adaptive hierarchical clustering~\cite{EK18}.

\subsection{Our contribution\label{our-results-subsec}}

In this paper we systematically investigate the problem of detecting one or
multiple hidden targets in a graph. Our work is driven by the open questions
posed by the recent paper of Emamjomeh-Zadeh et al.~\cite{EKS16} which dealt
with the detection of a single target with and without ``noise''. 
More specifically, Emamjomeh-Zadeh et al.~\cite%
{EKS16} asked for further fundamental generalizations of the model which
would be of interest, namely (a)~detecting a single target when the query
response is a direction to an \emph{approximately shortest path},
(b)~detecting a single target when querying a vertex that is an 
\emph{approximate median} of the current candidates' set $S$ (instead of an exact
one), and (c)~detecting \emph{multiple targets}, for which to the best of our
knowledge no progress has been made so far.

We resolve question (a) in Section~\ref{approx-shortest-subsec} by proving
that \emph{any} algorithm requires $\Omega (n)$ queries to detect a single
target $t$, assuming that a query directs to a path with an approximately
shortest length to $t$. Our results hold essentially for any approximation
guarantee, i.e.~for $1$-additive and for $(1+\varepsilon )$-multiplicative
approximations.

Regarding question (b), we first prove in Section~\ref{apx-median-subsec}
that, for any constant $0<\varepsilon <1$, the algorithm of~\cite{EKS16}
requires at least $\Omega(\sqrt{n})$ queries when we query each time an 
$(1+\varepsilon )$-approximate median (i.e.~an $(1+\varepsilon )$-approximate
minimizer of the potential $\Phi $ over the candidates' set $S$). Second, to
resolve this lower bound, we introduce in Section~\ref{our-potential-subsec}
a new potential $\Gamma $. This new potential can be efficiently computed
and, in addition, guarantees that, for any constant $0 \leq \varepsilon <1$,
the target $t$ can be detected in $O(\log n)$ queries even when an $%
(1+\varepsilon )$-approximate median (with respect to $\Gamma $) is queried
each time.

Regarding question (c), we initiate in Section~\ref{sec:two-targets} the
study for detecting multiple targets on graphs by focusing mainly to the case
of two targets $t_{1}$ and $t_{2}$. We assume throughout that every query
provides a correct answer, in the sense that it always returns a direction
to a shortest path from the queried vertex either to $t_{1}$ or to $t_{2}$.
The ``noise'' in this case is that the
algorithm does not know whether a query is returning a direction to $t_{1}$
or to $t_{2}$. Initially we observe in Section~\ref{sec:two-targets} that 
\emph{any} algorithm requires $\frac{n}{2}-1$ (resp.~$n-2$) queries in the
worst case to detect one target (resp.~both targets) if each query directs
adversarially to one of the two targets. Hence, in the remainder of Section~%
\ref{sec:two-targets}, we consider the case where each query independently
directs to the first target $t_{1}$ with a constant probability $p_{1}$ and
to the second target $t_{2}$ with probability $p_{2}=1-p_{1}$. For the case
of trees, we prove in Section~\ref{sec:two-targets} that both targets can be
detected with high probability within $O(\log^2 n)$ queries.

For general graphs, we distinguish between \emph{biased} queries ($%
p_{1}>p_{2}$) in Section~\ref{biased-subsec} and \emph{unbiased} queries ($%
p_{1}=p_{2}=\frac{1}{2}$) in Section~\ref{unbiased-subsec}. 
For biased queries we prove positive results, while for unbiased queries we derive strong negative results. 
For biased queries, first we observe that we can utilize 
the algorithm of Emamjomeh-Zadeh et al.~\cite{EKS16} to detect the 
first target $t_{1}$ with high probability in $O(\log n)$ queries; 
this can be done by considering the queries that direct to $t_{2}$ as ``noise''. 
Thus our objective becomes to detect the target $t_{2}$ in a polylogarithmic number of queries. 
Notice here that we cannot apply the ``noisy'' framework of~\cite{EKS16} to detect the second
target $t_{2}$ , since now the ``noise'' is larger than $\frac{1}{2}$. 
We prove our positive results for biased queries by making the additional assumption that, 
once a query at a vertex $v$ has chosen which target among $\{t_1,t_2\}$ it directs to, 
it returns any of the possible correct answers (i.e.~any of the neighbors $u$ of $v$ 
such that there exists a shortest path from $v$ to the chosen target using the edge $vu$) 
equiprobably and independently from all other queries. 
We derive a probabilistic algorithm that
overcomes this problem and detects the target $t_{2}$ with high probability
in $O(\Delta \log ^{2}n)$ queries, where $\Delta $ is the maximum degree of
a vertex in the graph. Thus, whenever $\Delta =O(poly\log n)$, a
polylogarithmic number of queries suffices to detect $t_{2}$.

In contrast, we prove in Section~\ref{unbiased-subsec} that, for unbiased queries, 
\emph{any} deterministic (possibly adaptive) algorithm that detects at least one
of the targets requires at least $\frac{n}{2}-1$ queries, even in an
unweighted cycle. Extending this lower bound for two targets, we prove that,
assuming $2c\geq 2$ different targets and unbiased queries, \emph{any}
deterministic (possibly adaptive) algorithm requires at least $\frac{n}{2}-c$
queries to detect one of the targets.

Departing from the fact that our best upper bound on the number of biased
queries in Section~\ref{biased-subsec} is not polylogarithmic when the
maximum degree $\Delta $ is not polylogarithmic, we investigate in 
Section~\ref{sec:other-queries} several variations of queries that provide more
informative responses. In Section~\ref{direction-distance-subsec} we turn
our attention to ``direction-distance'' biased queries which return with probability $p_{i}$ both the direction to a
shortest path to $t_{i}$ and the distance between the queried vertex and $%
t_{i}$. In Section~\ref{direction-edge-subsec} we consider another type of a
biased query which combines the classical ``direction'' query and an edge-variation of it. For both
query types of Sections~\ref{direction-distance-subsec} and~\ref%
{direction-edge-subsec} we prove that the second target $t_{2}$ can be
detected with high probability in $O(\log ^{3}n)$ queries. Furthermore, in
Sections~\ref{two-direction-subsec} and~\ref{restricted-queries-subsec} we
investigate two further generalizations of the ``direction'' query which make the target detection problem
trivially hard and trivially easy to solve, respectively.


\subsection{Our Model and Notation\label{our-model-subsec}}

We consider connected, simple, and undirected graphs. A graph $G=(V,E)$,
where $|V|=n$, is given along with a \emph{weight function} $w:E\rightarrow 
\mathbb{R}^{+}$ on its edges; if $w(e)=1$ for every $e\in E$ then $G$ is 
\emph{unweighted}. An edge between two vertices $v$ and $u$ of $G$ is
denoted by $vu$, and in this case $v$ and $u$ are said to be \emph{adjacent}. 
The distance $d(v,u)$ between vertices $v$ and $u$ is the length of a
shortest path between $v$ and $u$ with respect to the weight function $w$.
Since the graphs we consider are undirected, $d(u,v)=d(v,u)$ for every pair
of vertices $v,u$. Unless specified otherwise, all logarithms are taken with
base $2$. Whenever an event happens with probability at least $1-\frac{1}{%
n^{\alpha }}$ for some $\alpha >0$, we say that it happens \emph{with high
probability}.

The \emph{neighborhood} of a vertex $v\in V$ is the set $N(v)=\{u\in V:vu\in
E\}$ of its adjacent vertices. The cardinality of~$N(v)$ is the \emph{degree}~$\deg (v)$ of~$v$. 
The maximum degree among all vertices in $G$ is denoted
by $\Delta (G)$, i.e.~$\Delta (G)=\max \{\deg (v):v\in V\}$. For two
vertices $v$ and $u\in N(v)$ we denote by $N(v,u)=\{x\in
V:d(v,x)=w(vu)+d(u,x)\}$ the set of vertices $x\in V$ for which there exists
a shortest path from $v$ to $x$, starting with the edge $vu$. Note that, in
general, $N(u,v)\neq N(v,u)$. Let $T=\{t_{1},t_{2},\cdots
,t_{|T|}\}\subseteq V$ be a set of (initially unknown) \emph{target vertices}. 
A \emph{direction query} (or simply \emph{query}) at vertex $v\in V$
returns with probability $p_{i}$ a neighbor $u\in N(v)$ such that $t_{i}\in
N(u,v)$, where $\sum_{i=1}^{|T|}p_{i}=1$. If there exist more than one such
vertices $u\in N(v)$ leading to $t_{i}$ via a shortest path, the direction
query returns an arbitrary one among them, i.e.~possibly chosen
adversarially, unless specified otherwise. Moreover, if the queried vertex $%
v $ is equal to one of the targets $t_{i}\in T$, this is revealed by the
query with probability $p_{i}$.

\section{Detecting a Unique Target\label{one-target-sec}}

In this section we consider the case where there is only one unknown target $%
t=t_{1}$, i.e.~$T=\{t\}$. In this case the direction query at vertex $v$
always returns a neighbor $u\in N(v)$ such that $t\in N(v,u)$. For this
problem setting, Emamjomeh-Zadeh et al.~\cite{EKS16} provided a
polynomial-time algorithm which detects the target $t$ in at most $\log n$
direction queries. During its execution, the algorithm of~\cite{EKS16}
maintains a ``candidates' set'' $S\subseteq
V$ such that always $t\in S$, where initially $S=V$. At every iteration the
algorithm computes in polynomial time a vertex $v$ (called the \emph{median}
of $S$) which minimizes a potential $\Phi _{S}(v)$ among all vertices of the
current set $S$. Then it queries a median $v$ of $S$ and it reduces the
candidates' set $S$ to $S\cap N(v,u)$, where $u$ is the vertex returned by
the direction query at $v$. The upper bound $\log n$ of the number of
queries in this algorithm follows by the fact that always $|S\cap
N(v,u)|\leq \frac{|S|}{2}$, whenever $v$ is the median of~$S$.

\subsection{ Bounds for Approximately Shortest Paths\label%
{approx-shortest-subsec}}

We provide lower bounds for both additive and multiplicative approximation
queries. A $c$-\emph{additive approximation query} at vertex $v\in V$
returns a neighbor $u\in N(v)$ such that $w(vu) + d(u,t) \leq d(v,t) + c$.
Similarly, an $(1+\varepsilon)$-\emph{multiplicative approximation query} at
vertex $v\in V$ returns a neighbor $u\in N(v)$ such that $w(vu) + d(u,t)
\leq (1+\varepsilon) \cdot d(v,t)$.

It is not hard to see that in the unweighted clique with $n$ vertices any
algorithm requires in worst case $n-1$ $1$-additive approximation queries to
detect the target $t$. Indeed, in this case $d(v,t)=1$ for every vertex $v
\neq t$, while every vertex $u\notin \{v,t\}$ is a valid response of an 
$1$-additive approximation query at $v$. Since in the case of the unweighted
clique an additive 1-approximation is the same as a multiplicative
2-approximation of the shortest path, it remains unclear whether 
$1$-additive approximation queries allow more efficient algorithms for graphs
with large diameter. In the next theorem we strengthen this result to graphs
with unbounded diameter.

\begin{theorem}
\label{thm:approx-sp-add} Assuming $1$-additive approximation queries, any
algorithm requires at least $n-1$ queries to detect the target $t$, even in
graphs with unbounded diameter.
\end{theorem}

\begin{proof}
To prove the theorem we will construct a graph and a strategy for the
adversary such that any algorithm will need $n-1$ queries to locate the
target~$t$. Consider a horizontal $2 \times \frac{n}{2}$ grid graph where we
add the two diagonals in every cell of the grid. Formally, the graph has $%
\frac{n}{2}$ ``top'' vertices $v_1,\ldots,v_{\frac{n}{2}}$ and $\frac{n}{2}$
``bottom'' vertices $u_1,\ldots,u_{\frac{n}{2}}$. For every $i \in
\{1,2,\ldots,\frac{n}{2}-1\}$ we have the edges $v_{i}v_{i+1}, u_{i}u_{i+1},
v_{i}u_{i}, v_{i+1}u_{i+1}, v_{i}u_{i+1}, v_{i+1}u_{i}$.

The strategy of the adversary is as follows. If the algorithm queries a top
vertex $v_{i}$, then the query returns the bottom vertex $u_{i}$. Similarly,
if the algorithm queries a bottom vertex $u_{i}$, then the query returns the
top vertex $v_{i}$. Observe that, in every case, the query answer lies on a
path of length at most one more than a shortest path from the queried vertex
and the target $t$. To see this assume that the algorithm queries a top
vertex $v_i$; the case where the queried vertex is a bottom vertex $u_{i}$
is symmetric.

If $t=u_i$, then the edge $v_{i}u_{i}$ clearly lies on the shortest path
between $v_{i}$ and $t$. If $t=u_{j}$, where $j\neq i$, then the shortest
path uses one of the diagonal edges incident to $v_{i}$. In this case the
edge $v_{i}u_{i}$ leads to a path with length one more than the shortest
one. Finally, if $t=v_{j}$, where $j\neq i$, then the shortest path has
length $|j-i|$ and uses either the edge $v_iv_{i-1}$ or the edge $v_iv_{i+1}$%
. In both cases the edge $v_{i}u_{i}$ lies on the path from $v_{i}$ to $T$
with length $|j-i|+1$ which uses the edge $v_{i}u_{i}$ and one of the
diagonal edges $u_{i+1}v_{i-1}$ and $u_{i+1}v_{i+1}$.

Hence, after each query at a vertex different than $t$, the algorithm can
not obtain any information about the position of $t$ 
(except the fact that it is not the queried node). Thus, in the worst
case the algorithm needs to make $n-1$ queries to detect $t$.
\end{proof}

\medskip

In the next theorem we extend Theorem~\ref{thm:approx-sp-add} by showing a
lower bound of $n\cdot\frac{\varepsilon}{4}$ queries when we assume $%
(1+\varepsilon)$-multiplicative approximation queries.

\begin{theorem}
\label{thm:approx-sp-mul} Let $\varepsilon >0$. Assuming $(1+\varepsilon )$-multiplicative approximation queries, any algorithm requires at least at
least $n\cdot \frac{\varepsilon }{4}$ queries to detect the target $t$.
\end{theorem}

\begin{proof}
For the proof we use the same construction from Theorem~\ref{thm:approx-sp-add}, 
however the adversary we use here is slightly modified.
Assume that the distance between the queried vertex and the target $t$ is $d$. 
If $d+1 \leq (1+\varepsilon)\cdot d$, or equivalently, if $d \geq \frac{1}{\varepsilon}$, 
the adversary can respond in the same way as in Theorem~\ref{thm:approx-sp-add}.

Overall, the adversary proceeds as follows. Initially all vertices are
unmarked. Whenever the algorithm queries a vertex $v_{i}$ (resp.~$u_{i}$),
the adversary marks the vertices $\{v_{j},u_{j}:|j-i| < \frac{1}{\varepsilon}%
\}$ in order to determine the query response. If at least one unmarked
vertex remains in the graph, then the query returns (similarly to Theorem~%
\ref{thm:approx-sp-add}) vertex $u_{i}$ (resp.~$v_{i}$). In this case the
adversary can place the target $t$ at any currently unmarked vertex. By
doing so, the adversary ensures that the distance between $t$ and any of the
previously queried vertices is at least $\frac{1}{\varepsilon}$. If all
vertices of the graph have been marked, then the adversary places the target 
$t$ at one of the last marked vertices and in this case the query returns a
vertex on the shortest path between $t$ and the queried vertex.

With the above strategy, any algorithm needs to continue querying vertices
until there is no unmarked vertex left. Thus, since at every query the
adversary marks at most $2/\varepsilon$ new vertices, any algorithm needs to
perform at least $\frac{n/2}{2/\varepsilon}=n\cdot\frac{\varepsilon}{4}$
queries.
\end{proof}


\subsection{Lower Bound for querying the Approximate Median\label%
{apx-median-subsec}}

The potential $\Phi _{S}:V\rightarrow \mathbb{R}^{+}$ of~\cite{EKS16}, where 
$S\subseteq V$, is defined as follows. For any set $S\subseteq V$ and any
vertex $v\in V$, the potential of $v$ is $\Phi _{S}(v)=\sum_{u\in S}d(v,u)$.
A vertex $x\in V$ is an $(1+\varepsilon )$-approximate minimizer for the
potential $\Phi $ over a set $S$ (i.e.~an $(1+\varepsilon )$-median of~$S$)
if $\Phi _{S}(x)\leq (1+\varepsilon )\min_{v\in V}\Phi _{S}(v)$, where $%
\varepsilon >0$. We prove that an algorithm querying at each iteration
always an $(1+\varepsilon )$-median of the current candidates' set~$S$ needs 
$\Omega (\sqrt{n})$ queries.

\begin{theorem}
\label{thm:approx-eks-lb} Let $\varepsilon >0$. If the algorithm of~\cite{EKS16} 
queries at each iteration an $(1+\varepsilon )$-median for the
potential function $\Phi $, then at least $\Omega(\sqrt{n})$ queries are
required to detect the target $t$ in a graph $G$ with $n$ vertices, even if
the graph $G$ is a tree.
\end{theorem}

\begin{proof}
We will construct a graph $G=(V,E)$ with $n+1$ vertices such that $\Omega(%
\sqrt{n})$ queries are needed to locate the target. The graph $G$ will be a
tree with a unique vertex of degree greater than 2, i.e.~$G$ is a tree that
resembles the structure of a star. Formally, $G$ consists of $\sqrt{n}$
paths of length $\sqrt{n}$ each, where all these paths have a vertex $v_{0}$
as a common endpoint. Let $P_i = (v_0, v_{i,1}, v_{i,2}, \ldots, v_{i,\sqrt{n%
}-1}, v_{i,\sqrt{n}})$ be the $i$th path of $G$. For every $i\leq \sqrt{n}$
denote by $Q_{i} = \{v_{i,2},v_{i,3}, \ldots, v_{i,\sqrt{n}}\}$ be the set
of vertices of $P_{i}$ without $v_{0}$ and~$v_{i,1}$. Furthermore, for every 
$k \in \{0,1,\ldots, \sqrt{n}\}$ define $V_{-k} = V \setminus
(\bigcup_{1\leq i\leq k}Q_{i})$ to be the set of vertices left in the graph
by keeping only the first edge from each path $P_{i}$, where $i\leq k$. Note
by definition that $V_{-0}=V$.

Now recall that the algorithm of~\cite{EKS16} queries at each step an \emph{arbitrary median} 
for the potential function $\Phi$. 
To prove the theorem, it suffices to show that, in the graph $G$ that we constructed above, 
the slight modification of the algorithm of~\cite{EKS16} in which we query at each step 
an \emph{arbitrary $(1+\varepsilon )$-median} for the potential function $\Phi$, 
we need at least $\Omega(\sqrt{n})$ queries to detect the target in worst case. 
To this end, consider the target being vertex $v_{0}$. 
The main idea for the remainder of the proof is as follows. 
At every iteration the central vertex $v_{0}$ and all its
neighbors, who have not yet been queried, are $(1+\varepsilon )$-medians,
while $v_{0}$ is the exact median for the potential $\Phi $ of~\cite{EKS16}.
For every $k\in \{0,1,\ldots ,\sqrt{n}\}$ we have%
\begin{align}
\Phi _{V_{-k}}(v_{0})& =k+(\sqrt{n}-k)\sum_{j=1}^{\sqrt{n}}j  \notag \\
& =k+\frac{1}{2}(\sqrt{n}-k)(n+\sqrt{n}).  \label{eq1}
\end{align}

Next we compute $\Phi_{V_{-k}}(v_{p,1})$ for every $p > k$. Note that $%
d(v_{p,1},v_{i,1})=2$ for every $i \leq k$, and thus $\sum_{i=1}^k
d(v_{p,1},v_{i,1}) = 2k$. Furthermore, for the vertices on the path $P_p$ we
have 
\begin{equation*}
\sum_{i=2}^{\sqrt{n}} d(v_{p,1},v_{p,i}) = \sum_{i=1}^{\sqrt{n}-1}i = \frac{1%
}{2}(n - \sqrt{n}).
\end{equation*}
Finally, denote by $R=V_{-k} \setminus \{v_{p,1}, v_{p,2}, \ldots, v_{p,\sqrt{n}}\}$ 
the remaining of the vertices in $V_{-k}$. Then we have%

\begin{align*}
\sum_{u \in R} d(v_{p,1},u) & = 1 + (\sqrt{n}-k-1) \cdot \sum_{j=2}^{\sqrt{n}%
+1} j \\
& = 1 + \frac{1}{2}(\sqrt{n}-k-1)(n+3\sqrt{n}).
\end{align*}
Therefore, it follows that 
\begin{align}  \label{eq2}
\Phi_{V_{-k}}(v_{p,1}) = 2k + \frac{1}{2}(n - \sqrt{n}) + 1 + \frac{1}{2}(%
\sqrt{n}-k-1)(n+3\sqrt{n}).
\end{align}

Now note that, due to symmetry, $v_0$ is the exact median of the vertex set~$%
V$ (with respect to the potential $\Phi$ of~\cite{EKS16}), that is, $%
\Phi_{V}(v_0) = min_{x\in V}\{\Phi_{V}(x)\}$. Furthermore note by \eqref{eq1}
and \eqref{eq2} that $\Phi_{V_{-k}}(v_{p,1}) \geq \Phi_{V_{-k}}(v_0)$ for
every $k<\sqrt{n}$. Moreover, due to symmetry this monotonicity of $%
\Phi_{V_{-k}}(\cdot)$ is extended to all vertices $v_{p,2}, v_{p,3}, \ldots,
v_{p,\sqrt{n}}$, that is, $\Phi_{V_{-k}}(v_{p,j}) \geq \Phi_{V_{-k}}(v_0)$
for every $1\leq j\leq \sqrt{n}$. Therefore $v_0$ remains the exact median
of each of the vertex sets $V_{-k}$, where $0\leq k< \sqrt{n}$.

Let $\varepsilon>0$. Then \eqref{eq1} and \eqref{eq2} imply that $%
\Phi_{V_{-k}}(v_{p,1}) \leq (1+\varepsilon)\Phi_{V_{-k}}(v_0)$ for every $k
< \sqrt{n}$ and for large enough $n$. Now assume that the algorithm of~\cite%
{EKS16} queries always an $(1+\varepsilon)$-median of the candidates' set $S$%
, where initially $S=V$. Then the algorithm may query always a different
neighbor of $v_0$. Due to symmetry, we may assume without loss of generality
that the algorithm queries the vertices $v_{1,1}, v_{2,1}, \ldots, v_{\sqrt{n%
},1}$ in this order. Note that these vertices are $(1+\varepsilon)$-medians
of the candidates' sets $V_{-0}, V_{-1}, \ldots, V_{-(\sqrt{n}-1)}$,
respectively. Therefore the algorithm makes at least $\sqrt{n}$ queries,
where the total number of vertices in the graph is $n-\sqrt{n}+1$.
\end{proof}


\subsection{Upper Bound for querying the Approximate Median\label%
{our-potential-subsec}}

In this section we introduce a new potential function $\Gamma
_{S}:V\rightarrow \mathbb{N}$ for every $S\subseteq V$, which overcomes the
problem occured in Section~\ref{apx-median-subsec}. This new potential
guarantees efficient detection of $t$ in at most $O(\log n)$ queries, even
when we always query an $(1+\varepsilon )$-median of the current candidates'
set $S$ (with respect to the new potential $\Gamma $), for any constant $%
0<\varepsilon <1$. Our algorithm is based on the approach of~\cite{EKS16},
however we now query an approximate median of the current set $S$ with
respect to $\Gamma $ (instead of an exact median with respect to $\Phi $ of~\cite{EKS16}).

\begin{definition}[\ Potential $\Gamma$\ ]
\label{gamma-potential-def} Let $S\subseteq V$ and $v\in V$. Then $\Gamma
_{S}(v)=\max \{|N(v,u)\cap S|:u\in N(v)\}$.
\end{definition}

\begin{theorem}
\label{thm:approx-gamma-potential}Let $0\leq\varepsilon <1$. There exists an
efficient adaptive algorithm which detects the target $t$ in at most $\frac{%
\log n}{1-\log (1+\varepsilon )}$ queries, by querying at each iteration an $%
(1+\varepsilon )$-median for the potential function $\Gamma $.
\end{theorem}

\begin{proof}
Our proof closely follows the proof of Theorem 3 of~\cite{EKS16}. Let $%
S\subseteq V$ be an arbitrary set of vertices of $G$ such that $t\in S$. We
will show that there exists a vertex $v\in V$ such that $\Gamma _{S}(v)\leq 
\frac{|S|}{2}$. First recall the potential $\Phi _{S}(v)=\sum_{x\in S}d(v,x)$%
. Let now $v_{0}\in V$ be a vertex such that $\Phi _{S}(v_{0})$ is
minimized, i.e.~$\Phi _{S}(v_{0})\leq \Phi _{S}(v)$ for every $v\in V$. Let $%
u\in N(v_{0})$ be an \emph{arbitrary} vertex adjacent to $v_{0}$. We will
prove that $|N(v_{0},u)\cap S|\leq \frac{|S|}{2}$. Denote $%
S^{+}=N(v_{0},u)\cap S$ and $S^{-}=S\setminus S^{+}$. By definition, for
every $x\in S^{+}$, the edge $v_{0}u$ lies on a shortest path from $v_{0} $
to $x$, and thus $d(u,x)=d(v_{0},x)-w(v_{0}u)$. On the other hand, trivially 
$d(u,x)\leq d(v_{0},x)+w(v_{0}u)$ for every $x\in S$, and thus in particular
for every $x\in S^{-}$. Therefore $\Phi _{S}(v_{0})\leq \Phi _{S}(u)\leq
\Phi _{S}(v_{0})+(|S^{-}|-|S^{+}|)\cdot w(v_{0}u)$, and thus $|S^{+}|\leq
|S^{-}|$. That is, $|N(v_{0},u)\cap S|=|S^{+}|\leq \frac{|S|}{2}$, since $%
S^{-}=S\setminus S^{+}$. Therefore which then implies that $\Gamma
_{S}(v_{0})\leq \frac{|S|}{2}$ as the choice of the vertex $u\in N(v_{0})$
is arbitrary.

Let $v_{m}\in V$ be an exact median of $S$ with respect to $\Gamma $. That
is, $\Gamma _{S}(v_{m})\leq \Gamma _{S}(v)$ for every $v\in V$. Note that $%
\Gamma _{S}(v_{m})\leq \Gamma _{S}(v_{0})\leq \frac{|S|}{2}$. Now let $%
0\leq\varepsilon <1$ and let $v_{a}\in V$ be an $(1+\varepsilon )$-median of 
$S$ with respect to $\Gamma $. Then $\Gamma _{S}(v_{a})\leq (1+\varepsilon
)\Gamma _{S}(v_{m})\leq \frac{1+\varepsilon }{2}|S|$. Our adaptive algorithm
proceeds as follows. Similarly to the algorithm of~\cite{EKS16} (see Theorem
3 of~\cite{EKS16}), our adaptive algorithm maintains a candidates' set $S$,
where initially $S=V$. At every iteration our algorithm queries an arbitrary 
$(1+\varepsilon )$-median $v_{m}\in V$ of the current set $S$ with respect
to the potential $\Gamma $. Let $u\in N(v_{m})$ be the vertex returned by
this query; the algorithm updates $S$ with the set $N(v,u)\cap S$. Since $%
\Gamma _{S}(v_{a})\leq \frac{1+\varepsilon }{2}|S|$ as we proved above, it
follows that the updated candidates' set has cardinality at most $\frac{%
1+\varepsilon }{2}|S|$. Thus, since initially $|S|=n$, our algorithm detects
the target $t$ after at most $\log _{\left( \frac{2}{1+\varepsilon }\right)
}n=\frac{\log n}{1-\log (1+\varepsilon )}$ queries.
\end{proof}

\medskip

Notice in the statement of Theorem~\ref{thm:approx-gamma-potential} that for 
$\varepsilon=0$ (i.e.~when we always query an exact median) we get an upper
bound of $\log n$ queries, as in this case the size of the candidates' set
decreases by a factor of at least $2$. Furthermore notice that the reason
that the algorithm of~\cite{EKS16} is not query-efficient when querying an $%
(1+\varepsilon )$-median is that the potential $\Phi _{S}(v)$ of~\cite{EKS16}
can become quadratic in $|S|$, while on the other hand the value of our
potential $\Gamma _{S}(v)$ can be at most $|S|$ by Definition~\ref%
{gamma-potential-def}, for every $S\subseteq V$ and every $v\in V$.
Furthermore notice that, knowing only the value $\Phi _{S}(v)$ for some
vertex $v\in V$ is not sufficient to provide a guarantee for the
proportional reduction of the set $S$ when querying $v$. In contrast, just
knowing the value $\Gamma _{S}(v)$ directly provides a guarantee that, if we
query vertex $v$ the set $S$ will be reduced by a proportion of $\frac{%
\Gamma _{S}(v)}{|S|}$, regardless of the response of the query. Therefore,
in practical applications, we may not need to necessarily compute an (exact
or approximate) median of $S$ to make significant progress.


\section{Detecting Two Targets\label{sec:two-targets}}

In this section we consider the case where there are two unknown targets $%
t_{1}$ and $t_{2}$, i.e.~$T=\{t_{1},t_{2}\}$. In this case the direction
query at vertex $v$ returns with probability $p_{1}$ (resp.~with probability 
$p_{2}=1-p_{1}$) a neighbor $u\in N(v)$ such that $t_{1}\in N(v,u)$ (resp.~$%
t_{2}\in N(v,u)$). Detecting more than one unknown targets has been raised
as an open question by Emamjomeh-Zadeh et al.~\cite{EKS16}, while to the
best of our knowledge no progress has been made so far in this direction.
Here we deal with both problems of detecting at least one of the targets and
detecting both targets. We study several different settings and derive both
positive and negative results for them. Each setting differs from the other
ones on the ``freedom'' the adversary has
on responding to queries, or on the power of the queries themselves. We will
say that the response to a query \emph{directs to} $t_{i}$, where $i\in
\{1,2\}$, if the vertex returned by the query lies on a shortest path
between the queried vertex and $t_{i}$.

It is worth mentioning here that, if an adversary would be free to
arbitrarily choose which $t_{i}$ each query directs to (i.e.~instead of
directing to $t_{i}$ with probability $p_{i}$), then any algorithm would
require at least $\lfloor \frac{n}{2}\rfloor $ (resp.~$n-2$) queries to
detect at least one of the targets (resp.~both targets), even when the graph
is a path. Indeed, consider a path $v_{1},\ldots ,v_{n}$ where $t_{1}\in
\{v_{1},\ldots ,v_{\lfloor \frac{n}{2}\rfloor }\}$ and $t_{2}\in
\{v_{\lfloor \frac{n}{2}\rfloor +1},\ldots ,v_{n}\}$. Then, for every $i\in
\{1,\ldots ,\lfloor \frac{n}{2}\rfloor \}$, the query at $v_{i}$ would
return $v_{i+1}$, i.e.~it would direct to $t_{2}$. Similarly, for every $%
i\in \{\lfloor \frac{n}{2}\rfloor +1,\ldots ,n\}$, the query at $v_{i}$
would return $v_{i-1}$, i.e.~it would direct to $t_{1}$. It is not hard to
verify that in this case the adversary could ``hide'' the target $t_{1}$ at any of the first $\lfloor 
\frac{n}{2}\rfloor $ vertices which is not queried by the algorithm and the
target $t_{2}$ on any of the last $n-\lfloor \frac{n}{2}\rfloor $ vertices
which is not queried. Hence, at least $\lfloor \frac{n}{2}\rfloor $ queries
(resp.~$n-2$ queries) would be required to detect one of the targets (resp.
both targets) in the worst case.

As a warm-up, we provide in the next theorem an efficient algorithm that
detects with high probability both targets in a tree using $O(\log ^{2}n)$
queries. 

\begin{theorem}
\label{thm:iuartrees}For any constant $0<p_{1}<1$, we can detect with
probability at least $\left( 1-\frac{\log n}{n}\right) ^{2}$ both targets in
a tree with $n$ vertices using $O(\log ^{2}n)$ queries.
\end{theorem}

\begin{proof}
Let $G=(V,E)$ be a tree on $n$ vertices and let $T=\{t_{1},t_{2}\}$ be the
two targets. The algorithm runs in two phases. In each phase it maintains a
candidates' set $S\subseteq V$ such that, with high probability, $S$
contains at least one of the yet undiscovered targets. At the beginning of
each phase $S=V$. Let without loss of generality $p_{1}\geq p_{2}$.
Furthermore let $\alpha =-\frac{1}{\log p_{1}}$; note that $\alpha \geq 1$.

The first phase of the algorithm proceeds in $\log n$ iterations, as
follows. At the beginning of the $i$th iteration, where $1\leq i\leq \log n$%
, the candidates' set is $S_{i}$; note that $S_{1}=V$ at the beginning of
the first iteration. Let $v_{i}$ be a median of $S_{i}$ (with respect to the
potential $\Gamma $ of Section~\ref{our-potential-subsec}). In the first
iteration we query the median $v_{1}$ of $V$ once; let $u_{1}$ be the
response of this query. Then we know that one of the two targets belongs to
the set $N(v_{1},u_{1})$, thus we compute the updated candidates' set $%
S_{2}=N(v_{1},u_{1})$. Furthermore, since $v_{1}$ was chosen to be a median
of $S_{1}$, it follows that $|S_{2}|\leq \frac{|S_{1}|}{2}=\frac{n}{2}$.

For each $i\geq 2$, the $i$th iteration proceeds as follows. We query the
median $v_{i}$ of the set $S_{i}$ for $\alpha \log n$ times. First assume
that at least one of these $\alpha \log n$ queries at $v_{i}$ directs to a
subtree of $v_{i}$ (within $S_{i}$) that does not contain the first median $%
v_{1}$ of $S_{1}=V$, and let $u_{i}^{\prime }$ be the response of that
query. Then we know that the subtree of $v_{i}$ (within $S_{i}$) which is
rooted at $u_{i}^{\prime }$ contains at least one of the targets that belong
to $S_{i}$. Thus we compute the updated candidates' set $S_{i+1}=S_{i}\cap
N(v_{i},u_{i}^{\prime })$, where again $|S_{i+1}|\leq \frac{|S_{i}|}{2}$.

Now assume that all of the $\alpha \log n$ queries at $v_{i}$ direct to the
subtree of $v_{i}$ that contains the median $v_{1}$ of the initial
candidates' set $S_{1}=V$. Let $u_{i}^{\prime \prime }$ be the (unique)
neighbor of $v_{i}$ in that subtree, that is, all $\alpha \log n$ queries at 
$v_{i}$ return the vertex $u_{i}^{\prime \prime }$. Then we compute the
updated candidates' set $S_{i+1}=S_{i}\cap N(v_{i},u_{i}^{\prime \prime })$,
where again $|S_{i+1}|\leq \frac{|S_{i}|}{2}$. In this case, the probability
that at least one of the targets of $S_{i}$ does not belong to the subtree
of $v_{i}$ (within $S_{i}$) which is rooted at $u_{i}^{\prime \prime }$ is
upper bounded by the probability $p_{1}^{\alpha \log n}$ that each of the $%
\alpha \log n$ queries at $v_{i}$ directs to a target that does not belong
to $S_{i}$. That is, with probability at least $1-p_{1}^{\alpha \log n}$, at
least one of the targets of $S_{i}$ (which we are looking for) belongs to
the subtree of $v_{i}$ (within $S_{i}$) rooted at $u_{i}^{\prime \prime }$.
Since at each iteration the size of the candidates' set decreases by a
factor of $2$, it follows that $|S_{\log n}|=1$. The probability that at
each of the $\log n$ iterations we maintained a target from the previous
candidates' set to the next one is at least $\left( 1-p_{1}^{\alpha \log
n}\right) ^{\log n}=\left( 1-\frac{1}{n}\right) ^{\log n}\geq 1-\frac{\log n%
}{n}$ by Bernoulli's inequality. That is, with probability at least $1-\frac{%
\log n}{n}$ we detect during the first phase one of the two targets in $\log
n$ iterations, i.e.~in $\alpha \log ^{2}n$ queries in total.

Let $t_{0}$ be the target that we detected during the first phase. In the
second phase we are searching for the other target $t_{0}^{\prime }\in
T\setminus \{t_{0}\}$. The second phase of the algorithm proceeds again in $%
\log n$ iterations, as follows. Similarly to the first phase, we maintain at
the beginning of the $i$th iteration, where $1\leq i\leq \log n$, a
candidates' set $S_{i}$ with median $v_{i}$, where $S_{1}=V$ at the
beginning of the first iteration.

For each $i\geq 1$, in the $i$th iteration of the second phase we query $%
\alpha \log n$ times the median $v_{i}$ of the set $S_{i}$. First assume
that at least one of these $\alpha \log n$ queries at $v_{i}$ directs to a
subtree of $v_{i}$ (within $S_{i}$) that does not contain the target $t_{0}$
that we detected in the first phase, and let $u_{i}^{\prime }$ be the
response of that query. Then we can conclude that the other target $%
t_{0}^{\prime }$ belongs to the set $N(v_{i},u_{i}^{\prime })$, thus we
compute the updated candidates' set $S_{i+1}=S_{i}\cap N(v_{i},u_{i}^{\prime
})$, where $|S_{i+1}|\leq \frac{|S_{i}|}{2}$.

Now assume that all of the $\alpha \log n$ queries at $v_{i}$ direct to the
subtree of $v_{i}$ that contains the target $t_{0}$. Let $u_{i}^{\prime
\prime }$ be the (unique) neighbor of $v_{i}$ in that subtree, that is, all $%
\alpha \log n$ queries at $v_{i}$ return the vertex $u_{i}^{\prime \prime }$%
. Then we compute the updated candidates' set $S_{i+1}=S_{i}\cap
N(v_{i},u_{i}^{\prime \prime })$, where again $|S_{i+1}|\leq \frac{|S_{i}|}{2%
}$. In this case, the probability that the undiscovered target $%
t_{0}^{\prime }$ does not belong to the subtree of $v_{i}$ (within $S_{i}$)
which is rooted at $u_{i}^{\prime \prime }$ is upper bounded by the
probability $p_{1}^{\alpha \log n}$ that each of the $\alpha \log n$ queries
at $v_{i}$ directs to $t_{0}$. That is, with probability at least $%
1-p_{1}^{\alpha \log n}$, the target $t_{0}^{\prime }$ belongs to the
subtree of $v_{i}$ (within $S_{i}$) rooted at $u_{i}^{\prime \prime }$.
Since at each iteration the size of the candidates' set decreases by a
factor of at least $2$, it follows that $|S_{\log n}|=1$. The probability
that at each of the $\log n$ iterations we maintained the target $%
t_{0}^{\prime }$ in the candidates' set is at least $\left( 1-p_{1}^{\alpha
\log n}\right) ^{\log n}\geq 1-\frac{\log n}{n}$. That is, with probability
at least $1-\frac{\log n}{n}$ we detect in $\alpha \log ^{2}n$ queries
during the second phase the second target $t_{0}^{\prime }$, given that we
detected the other target $t_{0}$ in the first phase.

Summarizing, with probability at least $\left( 1-\frac{\log n}{n}\right)
^{2} $ we detect both targets in $2\alpha \log ^{2}n$ queries.
\end{proof}

\medskip

Since in a tree both targets $t_{1},t_{2}$ can be detected with high
probability in $O(\log ^{2}n)$ queries by Theorem~\ref{thm:iuartrees}, we
consider in the remainder of the section arbitrary graphs instead of trees.
First we consider in Section~\ref{biased-subsec} \emph{biased} queries,
i.e.~queries with $p_{1}>\frac{1}{2}$.  
Second we consider in Section~\ref{unbiased-subsec} \emph{unbiased} queries,
i.e.~queries with $p_{1}=p_{2}=\frac{1}{2}$.

\subsection{Upper Bounds for Biased Queries\label{biased-subsec}}

In this section we consider biased queries which direct to $t_{1}$ with
probability $p_{1}>\frac{1}{2}$ and to~$t_{2}$ with probability $%
p_{2}=1-p_{1}<\frac{1}{2}$. As we can detect in this case the first target $%
t_{1}$ with high probability in $O(\log n)$ queries by using the
``noisy'' framework of~\cite{EKS16}, our
aim becomes to detect the second target $t_{2}$ with the fewest possible
queries, once we have already detected~$t_{1}$.

For every vertex $v$ and every $i\in \{1,2\}$, denote by $%
E_{t_{i}}(v)=\{u\in N(v):t_{i}\in N(v,u)\}$ the set of neighbors of $v$ such
that the edge $uv$ lies on a shortest path from $v$ to $t_{i}$. Note that
the sets $E_{t_{1}}(v)$ and $E_{t_{2}}(v)$ can be computed in polynomial
time, e.g.~using Dijkstra's algorithm. We assume that, once a query at
vertex $v$ has chosen which target $t_{i}$ it directs to, it returns each
vertex of $E_{t_{i}}(v)$ equiprobably and independently from all other
queries. Therefore, each of the vertices of $E_{t_{1}}(v)\setminus
E_{t_{2}}(v)$ is returned by the query at $v$ with probability $\frac{p_{1}}{%
|E_{t_{1}}(v)|}$, each vertex of $E_{t_{2}}(v)\setminus E_{t_{1}}(v)$ is
returned with probability $\frac{1-p_{1}}{|E_{t_{2}}(v)|}$, and each vertex
of $E_{t_{1}}(v)\cap E_{t_{2}}(v)$ is returned with probability $\frac{p_{1}%
}{|E_{t_{1}}(v)|}+\frac{1-p_{1}}{|E_{t_{2}}(v)|}$. We will show in Theorem %
\ref{thm:two-targets-bounded} that, under these assumptions, we detect the
second target $t_{2}$ with high probability in $O(\Delta \log ^{2}n)$
queries where $\Delta $ is the maximum degree of the graph.

The high level description of our algorithm (Algorithm~\ref{alg:two-targets}) 
is as follows. Throughout the algorithm we maintain a candidates' set $S$
of vertices in which $t_{2}$ belongs with high probability. Initially $S=V$.
In each iteration we first compute an (exact or approximate) median $v$ of $%
S $ with respect to the potential $\Gamma $ (see Section~\ref%
{our-potential-subsec}). Then we compute the set $E_{t_{1}}(v)$ (this can be
done as $t_{1}$ has already been detected) and we query $c\Delta \log n$
times vertex $v$, where $c=\frac{7(1+p_{1})^{2}}{p_{1}(1-p_{1})^{2}}$ is a
constant. Denote by $Q(v)$ the multiset of size $c\Delta \log n$ that
contains the vertices returned by these queries at $v$. If at least one of
these $O(\Delta \log n)$ queries at $v$ returns a vertex $u\notin
E_{t_{1}}(v)$, then we can conclude that $u\in E_{t_{2}}(v)$, and thus we
update the set $S$ by $S\cap N(v,u)$. Assume otherwise that all $O(\Delta
\log n)$ queries at $v$ return vertices of $E_{t_{1}}(v)$. Then we pick a
vertex $u_{0}\in N(v)$ that has been returned most frequently among the $%
O(\Delta \log n)$ queries at $v$, and we update the set $S$ by $S\cap
N(v,u_{0})$. As it turns out, $u_{0}\in E_{t_{2}}(v)$ with high probability.
Since we always query an (exact or approximate) median $v$ of the current
candidates' set $S$ with respect to the potential $\Gamma $, the size of $S$
decreases by a constant factor each time. Therefore, after $O(\log n)$
updates we obtain $|S|=1$. It turns out that, with high probability, each
update of the candidates' set was correct, i.e.~$S=\{t_{2}\}$. Since for
each update of $S$ we perform $O(\Delta \log n)$ queries, we detect $t_{2}$
with high probability in $O(\Delta \log ^{2}n)$ queries in total.

\begin{algorithm}[htb]
\caption{Given $t_{1}$, detect $t_{2}$ with high probability with $O(\Delta \log^2 n)$ queries} \label{alg:two-targets}
\begin{algorithmic}[1]

\STATE{$S \leftarrow V$; \ $c \leftarrow \frac{7(1+p_{1})^2}{p_{1}(1-p_{1})^2}$}

\WHILE{$|S|>1$}
     \STATE{Compute an (approximate) median $v$ of $S$ with respect to potential $\Gamma$; \ Compute $E_{t_{1}}(v)$}
     \STATE{Query $c \Delta \log n$ times vertex $v$; \ Compute the multiset $Q(v)$ of these query responses} \label{step:direction}
     \IF{$Q(v) \setminus E_{t_{1}}(v) \neq \emptyset$}
          \STATE{Pick a vertex $u \in Q(v) \setminus E_{t_{1}}(v)$ and set $S \leftarrow S \cap N(v,u)$} \label{step:det-update}
     \ELSE
          \STATE{Pick a most frequent vertex $u \in Q(v)$ and set $S \leftarrow S \cap N(v,u)$} \label{step:prob-update}
       \ENDIF
\ENDWHILE

\medskip

\RETURN{the unique vertex in $S$}
\end{algorithmic}
\end{algorithm}

Recall that every query at $v$ returns a vertex $u\in E_{t_{1}}(v)$ with
probability $p_{1}$ and a vertex $u\in E_{t_{2}}(v)$ with probability $%
1-p_{1}$. Therefore, for every $v\in V$ the multiset $Q(v)$ contains at
least one vertex $u\in E_{t_{2}}(v)$ with probability at least $%
1-p_{1}^{|Q(v)|}=1-p_{1}^{|c\Delta \log n|}$. In the next lemma we prove
that, every time we update $S$ using Step~\ref{step:prob-update}, the
updated set contains $t_{2}$ with high probability.

\begin{lemma}
\label{lem:prob-update} Let $S\subseteq V$ such that $t_{2}\in S$ and let $%
S^{\prime }=S\cap N(v,u)$ be the updated set at Step~\ref{step:prob-update}
of Algorithm~\ref{alg:two-targets}. Then $t_{2}\in S^{\prime }$ with
probability at least $1-\frac{2}{n}$.
\end{lemma}

\begin{proof}
First we define the vertex subset $\widehat{E}_{t_{2}}(v) = E_{t_{2}}(v) \cap E_{t_{1}}(v)$. 
Assume that Step~\ref{step:prob-update} of Algorithm~\ref{alg:two-targets} is executed; 
then note that $Q(v)\subseteq E_{t_{1}}(v)$, 
i.e.~the query always returns either a vertex of $E_{t_{1}}(v)\setminus E_{t_{2}}(v)$ 
or a vertex of $\widehat{E}_{t_{2}}(v)$. 
Given the fact that Step~\ref{step:prob-update} of Algorithm~\ref{alg:two-targets} is executed, 
note that each of the vertices of $E_{t_{1}}(v)\setminus E_{t_{2}}(v)$ is returned 
by a query with probability $\frac{p_{1}}{|E_{t_{1}}(v)|}$ 
and each of the vertices of $\widehat{E}_{t_{2}}(v)$ is returned with probability 
$\frac{p_{1}}{|E_{t_{1}}(v)|}+\frac{1-p_{1}}{|\widehat{E}_{t_{2}}(v)|}$. 
Observe that these probabilities are the expected frequencies for these vertices in $Q(v)$, 
given the fact that $Q(v)\subseteq E_{t_{1}}(v)$. 
To prove the lemma it suffices to show that, whenever $Q(v)\subseteq E_{t_{1}}(v)$, 
the most frequent element of $Q(v)$ belongs to $E_{t_{1}}(v)\cap E_{t_{2}}(v)$ with high probability.
To this end, let $\delta =\frac{1-p_{1}}{1+p_{1}}$ and $c=\frac{7(1+p_{1})^{2}}{p_{1}(1-p_{1})^{2}}$ be two constants. 
Note that, for the chosen value of $\delta $, the inequality $|\widehat{E}_{t_{2}}(v)|\leq |E_{t_{}}(v)|$ 
is equivalent to
\begin{equation}
(1+\delta )\frac{p_{1}}{|E_{t_{1}}(v)|}\leq(1-\delta )\left( \frac{p_{1}}{%
|E_{t_{1}}(v)|}+\frac{1-p_{1}}{|\widehat{E}_{t_{2}}(v)|}\right)  \label{delta-bound-eq}
\end{equation}

Let $u\in E_{t_{1}}(v)\setminus E_{t_{2}}(v)$, i.e.~the query at $v$ directs
to $t_{1}$ but not to $t_{2}$. We define the random variable $Z_{i}(u)$,
such that $Z_{i}(u)=1$ if $u$ is returned by the $i$-th query at $v$ and $%
Z_{i}(u)=0$ otherwise. Furthermore define $Z(u)=\sum_{i=1}^{c\Delta \log
n}Z_{i}(u)$. Since $\Pr (Z_{i}(u)=1)=\frac{p_{1}}{|E_{t_{1}}(v)|}$, it
follows that $E(Z(u))=c\Delta \log n\frac{p_{1}}{|E_{t_{1}}(v)|}$ 
by the linearity of expectation. 
Then, using Chernoff's bounds it follows that%
\begin{align}
\Pr (Z(u)\geq (1+\delta )E(Z(u)))& \leq \exp \left( -\frac{\delta ^{2}}{3}%
\frac{p_{1}}{|E_{t_{1}}(v)|}c\Delta \log n\right)  \notag \\
& \leq\exp \left( -2\delta ^{2}\frac{(1+p_{1})^{2}}{(1-p_{1})^{2}}\log n\right)
\notag \\
& =\exp \left( -2\log n\right)  =\frac{1}{n^{2}}.  \label{chernoff-eq-1}
\end{align}%

Thus (\ref{chernoff-eq-1}) implies that the probability that there exists at
least one $u\in E_{t_{1}}(v)\setminus E_{t_{2}}(v)$ such that $Z(u)\geq
(1+\delta )E(Z(u))$ is%
\begin{equation}
\Pr \left( \exists u\in E_{t_{1}}(v)\setminus E_{t_{2}}(v):Z(u)\geq
(1+\delta )\frac{p_{1}}{|E_{t_{1}}(v)|}\right) \leq(\Delta -1)\frac{1}{n^{2}}\leq \frac{1}{n}.  \label{chernoff-eq-1-b}
\end{equation}

Now let $u^{\prime }\in \widehat{E}_{t_{2}}(v)$. Similarly to the
above we define the random variable $Z_{i}^{\prime }(u^{\prime })$, such
that $Z_{i}^{\prime }(u^{\prime })=1$ if $u^{\prime }$ is returned by the $i$%
-th query at $v$ and $Z_{i}^{\prime }(u^{\prime })=0$ otherwise. Furthermore
define $Z^{\prime }(u^{\prime })=\sum_{i=1}^{c\Delta \log n}Z_{i}^{\prime
}(u^{\prime })$. Since $\Pr (Z_{i}^{\prime }(u^{\prime })=1)=\frac{p_{1}}{%
|E_{t_{1}}(v)|}+\frac{1-p_{1}}{|\widehat{E}_{t_{2}}(v)|}$, by the linearity of expectation 
it follows that $%
E(Z(u))=c\Delta \log n\left( \frac{p_{1}}{|E_{t_{1}}(v)|}+\frac{1-p_{1}}{%
|\widehat{E}_{t_{2}}(v)|}\right) $. 
Then we obtain similarly to (\ref{chernoff-eq-1}) that 

\begin{align}
\Pr (Z^{\prime }(u^{\prime })\leq (1-\delta )E(Z^{\prime }(u^{\prime })))&
\leq \exp \left( -\frac{\delta ^{2}}{2}\left( \frac{p_{1}}{|E_{t_{1}}(v)|}+%
\frac{1-p_{1}}{|\widehat{E}_{t_{2}}(v)|}\right) c\Delta \log n\right)  \notag \\
& \leq\exp \left( -3\delta ^{2}\frac{(1+p_{1})^{2}}{p_{1}(1-p_{1})^{2}}\log
n\right)  \notag \\
& \leq\exp \left( -3\log n\right)  \leq \frac{1}{n^{2}}.  \label{chernoff-eq-2}
\end{align}

Thus, it follows by the union bound and by (\ref{delta-bound-eq}), (\ref{chernoff-eq-1-b}), 
and (\ref{chernoff-eq-2}) that%
\begin{equation}
\Pr (\exists u\in E_{t_{1}}(v)\setminus E_{t_{2}}(v):Z(u)\geq Z^{\prime
}(u^{\prime }))\leq \frac{2}{n}.  \label{chernoff-eq-3}
\end{equation}

That is, the most frequent element of $Q(v)$ belongs to $\widehat{E}_{t_{2}}(v)$ with probability at least $1-\frac{2}{n}$. This completes the
proof of the lemma.
\end{proof}

\medskip

With Lemma~\ref{lem:prob-update} in hand we can now prove the main theorem
of the section.

\begin{theorem}
\label{thm:two-targets-bounded}
Assume that every query at a vertex $v$ directs to $t_1$ and to $t_2$ 
with probability $p_1>\frac{1}{2}$ and $p_2=1-p_1$, respectively. 
Furthermore, once a query at a vertex $v$ has chosen 
which target it directs to, it returns any of the possible correct
answers equiprobably and independently from all other queries. 
Then, given $t_{1}$, Algorithm~\ref{alg:two-targets} detects $%
t_{2} $ in $O(\Delta \log ^{2}n)$ queries with probability at least $(1-%
\frac{2}{n})^{O(\log n)}$.
\end{theorem}

\begin{proof}
Since we query at each iteration an $(1+\varepsilon )$-median for the
potential function $\Gamma $, recall by Theorem~\ref%
{thm:approx-gamma-potential} that after at most $\frac{\log n}{1-\log
(1+\varepsilon )}=O(\log n)$ iterations we will obtain $|S|=1$. Furthermore,
in every iteration the algorithm queries $c\Delta \log n$ times the $%
(1+\varepsilon )$-median of the current set, and thus the algorithm makes $%
O(\Delta \log ^{2}n)$ queries in total. Whenever the algorithm updates $S$
in Step~\ref{step:det-update} the target $t_{2}$ belongs to the updated set
with probability 1. Moreover, whenever the algorithm updates $S$ in Step~\ref%
{step:prob-update}, Lemma ~\ref{lem:prob-update} implies that the target $%
t_{2}$ belongs to the updated set with probability at least $(1-\frac{2}{n})$%
. Thus, the probability all the $O(\log n)$ updates of $S$ were correct,
i.e.~$t_{2}$ belongs to $S$ after each of the $O(\log n)$ updates, is at
least $(1-\frac{2}{n})^{O(\log n)}$.
\end{proof}

\medskip

Note by Theorem~\ref{thm:two-targets-bounded} that, whenever $\Delta
=O(poly\log n)$ we can detect both targets $t_{1}$ and $t_{2}$ in $%
O(poly\log n)$ queries. However, for graphs with larger maximum degree $%
\Delta $, the value of the maximum degree dominates any polylogarithmic
factor in the number of queries. The intuitive reason behind this is that,
for an (exact or approximate) median $v$ of the current set $S$, whenever $%
\deg (v)$ and $E_{t_{1}}(v)$ are large and $E_{t_{2}}(v)\subseteq
E_{t_{1}}(v)$, we can not discriminate with a polylogarithmic number of
queries between the vertices of $E_{t_{2}}(v)$ and the vertices of $%
E_{t_{1}}(v)\setminus E_{t_{2}}(v)$ with large enough probability. Although
this argument does not give any lower bound for the number of queries in the
general case (i.e.~when $\Delta $ is unbounded), it seems that more
informative queries are needed to detect both targets with polylogarithmic
queries in general graphs. We explore such more informative queries in
Section~\ref{sec:other-queries}.

\subsection{Lower Bounds for Unbiased Queries\label{unbiased-subsec}}

In this section we consider unbiased queries, i.e.~queries which direct to
each of the targets $t_{1},t_{2}$ with equal probability $p_{1}=p_{2}=\frac{1}{2}$. 
In this setting every query is indifferent between the two targets,
and thus the ``noisy'' framework of~\cite{EKS16} cannot be applied for detecting any of the two targets. 
In particular, in this section we generalize our study to the case of $2c\geq 2$ different 
targets $T=\{t_{1},t_{2},\ldots ,t_{2c}\}$, where the query to any vertex $v\notin T$ is
unbiased. That is, $p_{i}=\frac{1}{2c}$ for every $i\in \{1,2,\ldots ,2c\}$.
In the next theorem we prove that any deterministic (possibly adaptive) algorithm 
needs at least $\frac{n}{2}-c$ queries to detect one of the $2c$ targets.

\begin{theorem}
\label{thm:lb-c-targets}Suppose that there are $2c$ targets in the graph and
let $p_{i}=\frac{1}{2c}$ for every $i\in \{1,2,\ldots ,2c\}$. Then, any
deterministic (possibly adaptive) algorithm requires at least $\frac{n}{2}-c$
queries to locate at least one target, even in an unweighted cycle.
\end{theorem}

\begin{proof}
Let $T=\{t_{1},t_{2},\ldots ,t_{2c}\}$ be the set of targets. Again, let $G$ be
the unweighted cycle with $n=2k$ vertices $v_{0},v_{1},\ldots ,v_{2k-1}$.
For each $i\in \{1,2,\ldots ,c\}$ the targets $\{t_{i},t_{i+c}\}$ are placed
by the adversary on two anti-diametrical vertices of the cycle, i.e.~$%
t_{i}=v_{j}$ and $t_{i+c}=v_{j+k}$, for some $j\in \{0,1,\ldots ,2k-1\}$.
Thus, for any vertex $v_{x}\notin T$, the unbiased query at $v_{x}$ returns $%
v_{x-1}$ with probability $\frac{1}{2}$ and $v_{x+1}$ with probability $%
\frac{1}{2}$. That is, for each vertex $v_{x}\notin T$ the response of the
query at $v_{x}$ is exactly the same. Let $\mathcal{A}$ be a deterministic algorithm that
queries at most $k-c-1$ different vertices. Then there exist at least $c+1$
pairs $\{v_{i_{1}},v_{i_{1}+k}\},\{v_{i_{2}},v_{i_{2}+k}\},\ldots
,\{v_{i_{c}},v_{i_{c}+k}\}$ of anti-diametrical vertices such that none of
these vertices is queried by the algorithm. Then the adversary can place the 
$2c$ targets any $c$ of these $c+1$ pairs of anti-diametrical vertices,
without affecting the validity of the previous answers. Thus the algorithm $%
\mathcal{A}$ needs to query at least $k-c=\frac{n}{2}-c$ different vertices
to detect a target.
\end{proof}

\begin{corollary}
\label{cor:iuarcycle}Let $p_{1}=p_{2}=\frac{1}{2}$. Then any deterministic
(possibly adaptive) algorithm needs at least $\frac{n}{2}-1$ queries to
detect one of the two targets, even in an unweighted cycle.
\end{corollary}

\section{More Informative Queries for Two Targets\label{sec:other-queries}}

A natural alternative to obtain query-efficient algorithms for multiple 
targets, instead of restricting the maximum degree $\Delta $ of the graph
(see Section~\ref{biased-subsec}), is to consider queries that provide more
informative responses in general graphs. As we have already observed in
Section~\ref{biased-subsec}, it is not clear whether it is possible to
detect multiple targets with $O(poly\log n)$ \emph{direction queries}
in an arbitrary graph. In this section we investigate natural variations and
extensions of the direction query for multiple targets which we studied
in Section~\ref{sec:two-targets}.

\subsection{Direction-Distance Biased Queries\label%
{direction-distance-subsec}}

In this section we strengthen the direction query in a way that it also
returns the value of the distance between the queried vertex and one of the
targets. More formally, a \emph{direction-distance query} at vertex $v\in V$
returns with probability $p_{i}$ a pair $(u,\ell )$, where $u\in N(v)$ such
that $t_{i}\in N(u,v)$ and $d(v,t_{i})=\ell $. Note that here we impose
again that all $p_{i}$'s are constant and that $\sum_{i=1}^{|T|}p_{i}=1$,
where $T=\{t_{1},t_{2},\ldots ,t_{|T|}\}$ is the set of targets. We will say
that the response $(u,\ell )$ to a direction-distance query at vertex $v$ 
\emph{directs to} $t_{i}$ if $t_{i}\in N(v,u)$ and $\ell =d(v,t_{i})$.
Similarly to our assumptions on the direction query, whenever there exist
more than one such vertices $u\in N(v)$ leading to $t_{i}$ via a shortest
path, the direction-distance query returns an arbitrary vertex $u$ among
them (possibly chosen adversarially). Moreover, if the queried vertex $v$ is
equal to one of the targets $t_{i}\in T$, this is revealed by the query with
probability $p_{i}$. These direction-distance queries have also been used in~%
\cite{EKS16} for detecting one single target in directed graphs.

Here we consider the case of two targets and \emph{biased queries}, i.e.~$%
T=\{t_{1},t_{2}\}$ where $p_{1}>p_{2}$. Similarly to Section~\ref%
{biased-subsec}, initially we can detect the first target $t_{1}$ with high
probability in $O(\log n)$ queries using the ``noisy'' model of~\cite{EKS16}. Thus, in what follows we
assume that $t_{1}$ has already been detected. We will show that the second
target $t_{2}$ can be detected with high probability with $O(\log ^{3}n)$
additional direction-distance queries using Algorithm~\ref{alg:dir-dist}.
The high level description of our algorithm is the following. We maintain a
candidates' set $S$ such that at every iteration $t_{2}\in S$ with high
probability. Each time we update the set $S$, its size decreases by a
constant factor. Thus we need to shrink the set $S$ at most $\log n$ times.
In order to shrink $S$ one time, we first compute an $(1+\varepsilon )$%
-median $v$ of the current set $S$ and we query $\log n$ times this vertex $%
v $. Denote by $Q(v)$ the set of all different responses of these $\log n$
direction-distance queries at $v$. As it turns out, the responses in $Q(v)$
might not always be enough to shrink $S$ such that it still contains $t_{2}$
with high probability. For this reason we also query $\log n$ times each of
the $\log n$ neighbors $u\in N(v)$, such that $(u,\ell )\in Q(v)$ for some $%
\ell \in \mathbb{N}$. After these $\log ^{2}n$ queries at $v$ and its
neighbors, we can safely shrink $S$ by a constant factor, thus detecting the
target $t_{2}$ with high probability in $\log ^{3}n$ queries.

For the description of our algorithm (see Algorithm~\ref{alg:dir-dist})
recall that, for every vertex $v$, the set $E_{t_{1}}(v)=\{u\in
N(v):t_{1}\in N(v,u)\}$ contains all neighbors of $v$ such that the edge $uv$
lies on a shortest path from $v$ to~$t_{1}$.

\begin{algorithm}[htb]
\caption{Given $t_{1}$, detect $t_{2}$ with high probability with $O(\log^3 n)$ direction-distance queries} \label{alg:dir-dist}
\begin{algorithmic}[1]

\STATE{$S \leftarrow V$}

\WHILE{$|S|>1$}\label{line-2-alg-2}
     \STATE{Compute an (approximate) median $v$ of $S$ with respect to potential $\Gamma$; \ Compute $E_{t_{1}}(v)$}
     \STATE{Query $\log n$ times vertex $v$; \ Compute the set $Q(v)$ of different query responses}
     \IF{there exists a pair $(u,\ell) \in Q(v)$ such that $u \notin E_{t_{1}}(v)$ or $\ell \neq d(v,t_{1})$} \label{step:det-update-3a}
          \STATE{$S \leftarrow S \cap N(v,u)$} \label{step:det-update-3}
     \ELSE
          \FOR{every $(u,\ell)\in Q(v)$} \label{step:det-else-for}
               \STATE{Query $\log n$ times vertex $u$; \ Compute the set $Q(u)$ of different query responses}
               
               \IF{for every $(z,\ell') \in Q(u)$ we have $\ell' = \ell - w(vu)$} \label{step:det-before-else-end}
                    \STATE{$S \leftarrow S \cap N(v,u)$; \ Goto line~\ref{line-2-alg-2}} \label{step:det-else-end}
               \ENDIF
          \ENDFOR
     \ENDIF
\ENDWHILE

\medskip

\RETURN{the unique vertex of $S$}
\end{algorithmic}
\end{algorithm}

In the next theorem we prove the correctness and the running time of
Algorithm~\ref{alg:dir-dist}.

\begin{theorem}
\label{thm:alg-dir-dis}Given $t_{1}$, Algorithm~\ref%
{alg:dir-dist} detects $t_{2}$ in at most $O(\log ^{3}n)$ queries with
probability at least $1-O\left( \log n\cdot p_{1}^{\log n}\right) $.
\end{theorem}

\begin{proof}
Throughout its execution, Algorithm~\ref{alg:dir-dist} maintains a vertex
set $S$ that contains the second target $t_{2}$ with high probability.
Initially $S=V$. Let $v$ be an $(1+\varepsilon )$-median of the set $S$
(with respect to the potential $\Gamma $ of Section~\ref%
{our-potential-subsec}) at some iteration of the algorithm, and assume that $%
t_{2}\in S$. We query $\log n$ times vertex $v$; let $Q(v)$ be the set of
all different query responses. Since each query directs to $t_{1}$ with
probability $p_{1}$ and to $t_{2}$ with probability $p_{2}$, it follows that
at least one of the queries at $v$ directs to $t_{2}$ with probability at
least $1-p_{1}^{\log n}$.

Consider a response-pair $(u,\ell )\in Q(v)$. If this query directs to $%
t_{1} $, then $u\in E_{t_{1}}(v)$ and $\ell =d(v,t_{1})$. Hence, if we
detect at least one response pair $(u,\ell )\in Q(v)$ such that $u\notin
E_{t_{1}}(v)$ or $\ell \neq d(v,t_{1})$, we can safely conclude that this
query directs to $t_{2}$ (lines~\ref{step:det-update-3a}-\ref%
{step:det-update-3} of Algorithm~\ref{alg:dir-dist}). Therefore, in this
case, $u\in E_{t_{2}}(v)=\{u\in N(v):t_{2}\in N(v,u)\}$, and thus we safely
compute the updated set $S\cap N(v,u)$ at line~\ref{step:det-update-3}.

Assume now that $u\in E_{t_{1}}(v)$ and $\ell =d(v,t_{1})$ for every
response-pair $(u,\ell )\in Q(v)$\ (see lines~\ref{step:det-else-for}-\ref%
{step:det-else-end} of the algorithm). Then every query at $v$ directs to $%
t_{1}$. However, as we proved above, at least one of these queries $(u,\ell
)\in Q(v)$ also directs to $t_{2}$ (i.e.~$u\in E_{t_{2}}(v)$) with
probability at least $1-p_{1}^{\log n}$. Therefore $\ell
=d(v,t_{1})=d(v,t_{2})$ with probability at least $1-p_{1}^{\log n}$. Note
that, in this case, we can not use only the response-pairs of $Q(v)$ to
distinguish which query directs to $t_{2}$.

In our attempt to detect at least one vertex $u\in E_{t_{2}}(v)$, we query $%
\log n$ times each of vertices $u$ such that $(u,\ell )\in Q(v)$. For each
such vertex $u$ denote by $Q(u)$ the set of all different response-pairs
from these $\log n$ queries at $u$. Similarly to the above, at least one of
these $\log n$ queries at $u$ directs to $t_{2}$ with probability at least $%
1-p_{1}^{\log n}$. Recall that $d(v,t_{2})=\ell $ and let $(z,\ell ^{\prime
})\in Q(u)$. If $u\in E_{t_{2}}(v)$ then $d(u,t_{2})=\ell -w(vu)$, otherwise 
$d(u,t_{2})>\ell -w(vu)$. Furthermore note that $d(u,t_{1})=\ell -w(vu)$,
since $u\in E_{t_{1}}(v)$. Therefore, if we detect at least one
response-pair $(z,\ell ^{\prime })\in Q(u)$ such that $\ell ^{\prime }>\ell
-w(vu)$, then we can safely conclude that $u\notin E_{t_{2}}(v)$. Otherwise,
if for every response-pair $(z,\ell ^{\prime })\in Q(u)$ we have that $\ell
^{\prime }=\ell -w(vu)$, then $u\in E_{t_{2}}(v)$ (i.e.~$t_{2}\in N(v,u)$) 
with probability at least $1-p_{1}^{\log n}$.

Recall that there exists at least one query at $v$ that directs to $t_{2}$
with probability at least $1-p_{1}^{\log n}$, as we proved above. That is,
among all response-pairs $(u,\ell )\in Q(v)$ there exists at least one
vertex $u\in E_{t_{2}}(v)$ with probability at least $1-p_{1}^{\log n}$.
Therefore, we will correctly detect a vertex $u\in E_{t_{2}}(v)$ at lines %
\ref{step:det-before-else-end}-\ref{step:det-else-end} of the algorithm with
probability at least $\left( 1-p_{1}^{\log n}\right) ^{2}$, i.e.~with at
least this probability the updated candidates' set at line~\ref%
{step:det-else-end} still contains $t_{2}$. Thus, since we shrink the
candidates' set $\frac{\log n}{1-\log (1+\varepsilon )}=O(\log n)$ times, we
eventually detect $t_{2}$ as the unique vertex in the final candidates' set
with probability at least $\left( 1-p_{1}^{\log n}\right) ^{O(\log n)}\geq
1-O(\log n\cdot p_{1}^{\log n})$ by Bernoulli's inequality. Finally, it is
easy to verify from the above that the algorithm will terminate after at
most $O(\log ^{3}n)$ queries with probability at least $1-O(\log n\cdot
p_{1}^{\log n})$.
\end{proof}

\subsection{Vertex-Direction and Edge-Direction Biased Queries\label%
{direction-edge-subsec}}

An alternative natural variation of the direction query is to \emph{query an
edge} instead of querying a vertex. More specifically, the direction query
(as defined in Section~\ref{our-model-subsec}) queries a vertex $v\in V$ and
returns with probability $p_{i}$ a neighbor $u\in N(v)$ such that $t_{i}\in
N(u,v)$. Thus, as this query always queries a vertex, it can be also
referred to as a \emph{vertex-direction query}. Now we define the \emph{%
edge-direction query} as follows: it queries an ordered pair of adjacent
vertices $(v,u)$ and it returns with probability $p_{i}$ \texttt{YES} (resp.~\texttt{NO}) if $t_{i}\in N(v,u)$ (resp.~if $t_{i}\notin N(v,u)$). Similarly
to our notation in the case of vertex-direction queries, we will say that
the response \texttt{YES} (resp.~\texttt{NO}) to an edge-direction query at
the vertex pair $(v,u)$ \emph{refers} to $t_{i}$ if $t_{i}\in N(v,u)$ (resp.
if $t_{i}\notin N(v,u)$). Similar but different edge queries for detecting
one single target on trees have been investigated in~\cite{EKS16, IRV88,
N09, S89}.

Here we consider the case where both vertex-direction and edge-direction
queries are available to the algorithm, and we focus again to the case of
two targets and \emph{biased queries}, i.e.~$T=\{t_{1},t_{2}\}$ where $%
p_{1}>p_{2}$. Similarly to Sections~\ref{biased-subsec} and~\ref%
{direction-distance-subsec}, we initially detect $t_{1}$ with high
probability in $O(\log n)$ vertex-direction queries using the
``noisy'' model of~\cite{EKS16}. Thus, in
the following we assume that $t_{1}$ has already been detected. 
 We will show that Algorithm~\ref{alg:vertex-edge-dir} detects the second target $t_{2}$
with high probability using $O(\log ^{2}n)$ additional vertex-direction
queries and $O(\log ^{3}n)$ edge--direction queries, i.e.~in total $O(\log ^{3}n)$ queries.

\begin{algorithm}[htb]
\caption{Given $t_{1}$, detect $t_{2}$ with high probability with $O(\log^3 n)$ vertex-direction and edge-direction queries} \label{alg:vertex-edge-dir}
\begin{algorithmic}[1]

\STATE{$S \leftarrow V$} \label{alg-3-line-1}

\WHILE{$|S|>1$} \label{alg-3-line-2}
     \STATE{Compute an (approximate) median $v$ of $S$ with respect to potential $\Gamma$; \ Compute $E_{t_{1}}(v)$} \label{alg-3-line-3}
     \STATE{Apply $\log n$ vertex-direction queries at vertex $v$; \ Compute the set $Q(v)$ of different query responses} \label{alg-3-line-4}
     \IF{there exists a vertex $u \in Q(v)$ such that $u \notin E_{t_{1}}(v)$} \label{alg-3-line-5}
          \STATE{$S \leftarrow S \cap N(v,u)$} \label{alg-3-line-6}
     \ELSE \label{alg-3-line-7}
          \FOR{every $u\in Q(v)$} \label{alg-3-line-8}
               \STATE{Apply $\log n$ edge-direction queries at $(v,u)$; \ Compute the set $Q(v,u)$ of different query responses}\label{alg-3-line-9}%
               \IF{$Q(v,u)=\{\text{\texttt{YES}}\}$} \label{alg-3-line-10}
                    \STATE{$S \leftarrow S \cap N(v,u)$; \ Goto line~\ref{alg-3-line-2}} \label{alg-3-line-11}
               \ENDIF
          \ENDFOR
     \ENDIF
\ENDWHILE

\medskip

\RETURN{the unique vertex of $S$}
\end{algorithmic}
\end{algorithm}

In the next theorem we prove the correctness and the running time of
Algorithm~\ref{alg:vertex-edge-dir}.

\begin{theorem}
\label{thm:alg-vertex-edge-dir}Given $t_{1}$, 
 Algorithm~\ref{alg:vertex-edge-dir} detects
$t_{2}$ in at most $O(\log ^{2}n)$
vertex-direction queries and $O(\log ^{3}n)$ edge--direction queries with
probability at least $1-O(\log n\cdot p_{1}^{\log n})$.
\end{theorem}

\begin{proof}
The proof follows a similar approach as the proof of Theorem~\ref%
{thm:alg-dir-dis}. Throughout its execution, Algorithm~\ref%
{alg:vertex-edge-dir} maintains a vertex set $S$ that contains the second
target $t_{2}$ with high probability. Initially $S=V$. Let $v$ be an $%
(1+\varepsilon )$-median of the set $S$ (with respect to the potential $%
\Gamma $ of Section~\ref{our-potential-subsec}) at some iteration of the
algorithm, and assume that $t_{2}\in S$. We query $\log n$ times vertex $v$;
let $Q(v)$ be the set of all different query responses. Similarly to the
analysis of Algorithm~\ref{alg:dir-dist} in the proof of Theorem~\ref%
{thm:alg-dir-dis}, at least one of the queries at $v$ directs to $t_{2}$
with probability at least $1-p_{1}^{\log n}$.

Consider a response-vertex $u\in Q(v)$. If this query directs to $t_{1}$,
then $u\in E_{t_{1}}(v)$. Hence, if we detect at least one $u\in Q(v)$ such
that $u\notin E_{t_{1}}(v)$, we can safely conclude that this query directs
to $t_{2}$ (lines~\ref{alg-3-line-5}-\ref{alg-3-line-6} of Algorithm~\ref%
{alg:vertex-edge-dir}). Therefore, in this case, $u\in E_{t_{2}}(v)=\{u\in
N(v):t_{2}\in N(v,u)\}$, and thus we safely compute the updated set $S\cap
N(v,u)$ at line~\ref{alg-3-line-6}.

Assume now that $u\in E_{t_{1}}(v)$ for every response $u\in Q(v)$ (see
lines~\ref{alg-3-line-8}-\ref{alg-3-line-11} of the algorithm). Then every
query at $v$ directs to $t_{1}$, although at least one of them also directs
to $t_{2}$ (i.e.~$Q(v)\cap E_{t_{2}}(v)\neq \emptyset $) with probability at
least $1-p_{1}^{\log n}$, as we proved above. Note that, in this case, we
can not use only the vertices of $Q(v)$ to distinguish which query directs
to $t_{2}$.

In our attempt to detect at least one vertex $u\in E_{t_{2}}(v)$, we apply $%
\log n$ edge-direction queries at each of the ordered pairs $(v,u)$, where $%
u\in Q(v)$. For each such pair $(v,u)$ denote by $Q(v,u)$ the set of all
different \texttt{YES}/\texttt{NO} responses from these $\log n$ queries at $%
(v,u)$. Similarly to the above, at least one of these $\log n$ queries at $%
(v,u)$ refers to $t_{2}$ with probability at least $1-p_{1}^{\log n}$.
Therefore, if \texttt{NO}$\in Q(v,u)$, then we can safely conclude that $%
u\notin E_{t_{2}}(v)$. Otherwise, if $Q(v,u)=\{\mathtt{YES}\}$, then $u\in
E_{t_{2}}(v)$ (i.e.~$t_{2}\in N(v,u)$) with probability at least $%
1-p_{1}^{\log n}$.

Recall that there exists at least one query at $v$ that directs to $t_{2}$
with probability at least $1-p_{1}^{\log n}$. That is, among all responses
in $Q(v)$ there exists at least one vertex $u\in E_{t_{2}}(v)$ with
probability at least $1-p_{1}^{\log n}$. Therefore, we will correctly detect
a vertex $u\in E_{t_{2}}(v)$ at lines~\ref{alg-3-line-10}-\ref{alg-3-line-11}
of the algorithm with probability at least $\left( 1-p_{1}^{\log n}\right)
^{2}$, i.e.~with at least this probability the updated candidates' set at
line~\ref{alg-3-line-11} still contains $t_{2}$. Thus, similarly to the
proof of Theorem~\ref{thm:alg-dir-dis}, we eventually detect $t_{2}$ as the
unique vertex in the final candidates' set with probability at least $%
1-O(\log n\cdot p_{1}^{\log n})$. Finally, it is easy to verify from the
above that the algorithm will terminate after at most $O(\log ^{2}n)$
vertex-direction queries and $\log ^{3}n$ edge--direction queries with
probability at least $1-O(\log n\cdot p_{1}^{\log n})$.
\end{proof}

\subsection{Two-Direction Queries\label{two-direction-subsec}}

In this section we consider another variation of the direction query that
was defined in Section~\ref{our-model-subsec} (or ``vertex-direction query'' in the terminology of Section~\ref%
{direction-edge-subsec}), which we call \emph{two-direction query}.
Formally, a two-direction query at vertex $v$ returns an \emph{unordered pair%
} of (not necessarily distinct) vertices $\{u,u^{\prime }\}$ such that $%
t_{1}\in N(v,u)$ and $t_{2}\in N(v,u^{\prime })$. Note here that, as $%
\{u,u^{\prime }\}$ is an unordered pair, the response of the two-direction
query does not clarify which of the two targets belongs to $N(v,u)$ and
which to $N(v,u^{\prime })$.

Although this type of query may seem at first to be more informative than
the standard direction query studied in Section~\ref{sec:two-targets}, we
show that this is not the case. Intuitively, this type of query resembles
the unbiased direction query of Section~\ref{unbiased-subsec}. To see this,
consider e.g.~the unweighted cycle where the two targets are placed at two
anti-diametrical vertices; then, applying many times the unbiased direction
query of Section~\ref{unbiased-subsec} at any specific vertex $v$ reveals
with high probability the same information as applying a single
two-direction query at $v$. Based on this intuition the next theorem can be
proved with exactly the same arguments as Theorem~\ref{thm:lb-c-targets} of
Section~\ref{unbiased-subsec}.

\begin{theorem}
\label{thm:both-directions}Any deterministic (possibly adaptive) algorithm
needs at least $\frac{n}{2}-1$ two-direction queries to detect one of the
two targets, even in an unweighted cycle.
\end{theorem}

\subsection{Restricted Set Queries\label{restricted-queries-subsec}}

The last type of queries we consider is when the query is applied not only
to a vertex $v$ of the graph, but also to a subset $S\subseteq V$ of the
vertices, and the response of the query is a vertex $u\in N(v)$ such that $%
t\in N(v,u)$ for at least one of the targets $t$ that belong to the set $S$.
Formally, let $T$ be the set of targets. The \emph{restricted-set query} at
the pair $(v,S)$, where $v\in V$ and $S\subseteq V$ such that $T\cap S\neq
\emptyset $, returns a vertex $u\in N(v)$ such that $t\in N(v,u)$ for at
least one target $t\in T\cap S$. If there exist multiple such vertices $u\in
N(v)$, the query returns one of them adversarially. Finally, if we query a
pair $(v,S)$ such that $T\cap S=\emptyset $, then the query returns
adversarially an arbitrary vertex $u\in N(v)$, regardless of whether the
edge $vu$ leads to a shortest path from $v$ to any target in $T$. That is,
the response of the query can be considered in this case as ``noise''.

In the next theorem we prove that this query is very powerful, as $|T|\cdot
\log n$ restricted-set queries suffice to detect all targets of the set $T$.

\begin{theorem}
\label{thm:restricted-set}Let $T$ be the set of targets. There exists an
adaptive deterministic algorithm that detects all targets of $T$ with at
most $|T|\cdot \log n$ restricted-set queries.
\end{theorem}

\begin{proof}
To detect the first target we simply apply binary search on graphs. At every
iteration we maintain a candidates' set $S$ (initially $S=V$). We compute a
median $v$ of $S$ (with respect to the potential $\Gamma $ of Section~\ref%
{our-potential-subsec}) and we query the pair $(v,S)$. If the response of
the query at $(v,S)$ is vertex $u\in N(v)$ then we update the candidates'
set as $S\cap N(v,u)$. We know that there is at least one target in the
updated set $S$ and that the size of the candidates' set decreased by a
factor of at least $2$ (cf. Theorem~\ref{thm:approx-gamma-potential}). Thus,
after at most $\log n$ restricted-set queries we end up with a candidates'
set of size 1 that contains one target.

We repeat this procedure for another $|T|-1$ times to detect all remaining
targets of $T$,as follows. Assume that we have already detected the targets $%
t_{1},t_{2},\ldots ,t_{i}\in T$. To detect the next target of $T$ we
initially set $S=V\setminus \{t_{1},t_{2},\ldots ,t_{i}\}$ and we apply the
above procedure. Then, after at most $\log n$ restricted-set queries we
detect the next target $t_{i+1}$. Thus, after at most $|T|\cdot \log n$
restricted-set queries in total we detect all targets of $T$.
\end{proof}

\section{Conclusions\label{conclusions-sec}}

This paper resolves some of the open questions raised by Emamjomeh-Zadeh et
al.~\cite{EKS16} and makes a first step towards understanding the query
complexity of detecting two targets on graphs. Our results provide evidence
that different types of queries can significantly change the difficulty of
the problem and make it from almost trivial impossible to solve.

The potential $\Gamma$ we introduced in this paper 
has several interesting properties that have not yet been fully explored. As we mentioned in the paper, 
just knowing the value $\Gamma_S(v)$ for a vertex $v$ directly provides enough information to 
quantify the ``progress'' a direction query can make by querying vertex $v$, 
without the need to know the values $\Gamma_S(u)$ for any other vertex $u\neq v$. 
This property of $\Gamma$ may be exploited to provide computationally more efficient algorithms 
for detecting one target; an algorithm might only need to compute $\Gamma_S(v)$ 
for all vertices $v$ lying within a wisely chosen subset such that one of these vertices 
is an approximate median. 
Of course, this approach cannot break the $\log n$ lower bound on the number of queries 
needed to detect the target (e.g.~in the path of $n$ vertices), but it could potentially improve the
computational complexity of the detection algorithm. Furthermore, the potential $\Gamma$ 
might be a useful tool for deriving an optimal number of queries for classes of
graphs other than trees, since every exact median of $\Gamma$ separates the graph 
into roughly equal subgraphs. 
By resolving an open question of~\cite{EKS16} we proved that, 
assuming that a query directs to a path with an approximately shortest path to the (single) target $t$, 
\emph{any} algorithm requires $\Omega(n)$ queries to detect $t$. 
It remains open to specify appropriate special graph classes (or other special conditions) 
that allow the detection of $t$ using a polylogarithmic number of such approximate-path queries.

For the setting where two, or more, targets need to be detected there is a
plethora of interesting questions. We believe that the most prominent one is to
derive lower bounds on the number of queries needed to detect both targets 
in the \emph{biased} setting. Can the number of queries be improved to $O(\log n)$, or 
$O(\log n \cdot poly \log \log(n))$? We have preliminary results that suggest a lower bound of 
$\log n \log\log n$ bound for a special type of algorithms, however a general lower bound
seems to require new techniques.
Another intriguing question is to find the minimal requirements a query
has to satisfy in order to detect even one target in the \emph{unbiased} setting. 
Furthermore, in the \emph{biased} setting, it is not completely clear whether all our assumptions 
in the statement of Theorem~\ref{thm:two-targets-bounded} are necessary to prove its correctness; 
however we believe they are. 
In particular, can we get in Theorem~\ref{thm:two-targets-bounded} an upper bound of $O(\Delta \log^2 n)$ 
biased queries for detecting the second target, 
if we assume that, whenever a query has chosen to direct to a specific target (with a biased probability), 
it directs to an \emph{adversarially} chosen correct answer? 
Is the dependence on $\Delta$ necessary, even if we assume (as in Theorem~\ref{thm:two-targets-bounded}) 
that a query randomly chooses among the correct answers?


\end{document}